\documentclass[12pt]{article}

\usepackage{epsfig}

\usepackage{amssymb}
   \usepackage{amsmath}

\usepackage{amsfonts}

\def\ha{\mbox{$\frac{1}{2}$}}
\def\be{\begin{equation}}
\def\ee{\end{equation}}
\def\ba{\begin{array}{c}}
\def\ea{\end{array}}

\def\ben{$$}
\def\een{$$}

\newcommand{\bea}{\begin{eqnarray}}
\newcommand{\eea}{\end{eqnarray}}

\newcommand{\kt}{\rangle}

\newtheorem{thm}{Theorem}
\newtheorem{cor}[thm]{Corollary}
\newtheorem{lemma}[thm]{Lemma}

\newtheorem{conj}[thm]{Conjecture}
\newenvironment{proof}{\noindent {\bf Proof.}}{\hfill$\square$\vspace{3mm}\endtrivlist}

\begin{document}

\titlepage

  \begin{center}{\Large \bf

Bound states emerging from below the continuum in a solvable ${\cal
PT}-$symmetric discrete Schr\"{o}dinger equation

 }\end{center}

\vspace{5mm}

  \begin{center}

{\bf Miloslav Znojil}

 \vspace{3mm}

Nuclear Physics Institute ASCR, 250 68 \v{R}e\v{z}, Czech
Republic\footnote{ e-mail: znojil@ujf.cas.cz}\\

 \vspace{3mm}

\end{center}

\vspace{5mm}

\section*{Abstract}

The phenomenon of the birth of an isolated quantum bound state at the
lower edge of the continuum
is studied for a particle moving along a discrete real line of
coordinates $x \in \mathbb{Z}$. The motion is controlled by a weakly
nonlocal $2J-$parametric external potential $V$ which is
non-Hermitian but ${\cal PT}-$symmetric. The model 
is found exactly solvable. The bound states are interpreted
as Sturmians. Their closed-form definitions are
presented and discussed up to $J=7$.

 \vspace{9mm}

%


 \begin{center}
\end{center}

 \newpage

\section{Introduction}

The current growth of interest in the mathematical aspects of
non-self-adjoint operators \cite{book,PHHQP} has one of its roots in
the increasing popularity of quantum physics of systems which
exhibit a combined parity plus time-reversal invariance {\it
alias\,} ${\cal PT}-$symmetry \cite{BB,Carl}. In this innovative
branch of physics one can encounter discontinuities in the
time-evolution processes \cite{Stefan} and the critical situations
in which a small change of a parameter $\lambda$ in a non-Hermitian
but ${\cal PT}-$symmetric Hamiltonian $H(\lambda)\neq
H^\dagger(\lambda)$ causes an abrupt change of the observable
properties of the system in question \cite{catast}. A broad class of
the latter phenomena could be called ``quantum catastrophes''. Their
mathematical background may be found explained in the Kato's
monograph on linear operators \cite{Kato}. The points of the sudden
change were given there the name of ``exceptional points'' (EP).

The localization of the singular EP values of
$\lambda=\lambda_{(EP)}$ originally helped to determine the radii of
the convergence of perturbation series. As long as the Hamiltonians
in question were usually assumed self-adjoint, their EP parameters
were complex, i.e., not carrying any immediate physical meaning
(cf.~Refs.~\cite{BW} for typical illustrative examples). The
paradigm has been recently changed due to a dramatic increase of
interest in non-Hermitian models. The EP singularities
$\lambda_{(EP)}$ can be real in these models~\cite{Alvarez}.
Traditionally, people were connecting their occurrence with the
physics of the so called ``open systems''. Such a scenario may be
found described in a number of reviews (cf., e.g.,
\cite{Moiseyev,Rotter}), emphasizing that the open systems are, as a
rule, unstable, resonant and quickly decaying or growing.

The comparatively recent interest in the ${\cal PT}-$symmetric
theory of reviews \cite{Carl,ali,SIGMA} has a perceivably different
motivation and, in particular, an entirely different perception of
the physical role of the real EP singularities. The difference is
reflected by the terminology (used, say, in nuclear physics
\cite{Geyer}) by which the non-Hermitian Hamiltonians of the stable
and unitary quantum systems are called quasi-Hermitian. This means
that they {\em share} many of their {\em observable} features with
the self-adjoint models of conventional textbooks.

One of the most characteristic distinctive features of
quasi-Hermitian Hamiltonians $H(\lambda)$ is that they resemble
their standard self-adjoint analogues merely inside certain domains
(i.e., open sets) ${\cal D}$ of the admissible (and, say, real)
values of the parameter or parameters $\lambda$. Naturally, the
localization of the boundaries ${\cal D}$ of these domains (i.e., of
the EP manifolds $\partial {\cal D}$) belongs to one of the main
mathematical tasks for theoreticians.

In the literature the subjects is being developed with due care.
There exist multiple model-based studies of ${\cal PT}-$symmetric
Hamiltonians $H(\lambda)$ in which the ${\cal PT}-$symmetry gets
spontaneously broken ({\it pars pro toto}, see
\cite{BB,sgezou,ambi}). This implies that at least some of the
bound-state energies become complex after the passage of $\lambda$
through an EP element $\lambda_{(EP)}$ of the boundary $\partial
{\cal D}$. In parallel, there also appeared a few studies
\cite{borisov,sdenisem,I} in which the authors described another
real-EP-related dynamical-evolution scenario. In it, the spontaneous
breakdown of ${\cal PT}-$symmetry appeared inhibited and replaced by
an instantaneous recovery of the symmetry (cf., e.g., the toy models
of Refs.~\cite{sdenisem,ptho,sdavidem} exhibiting the EP-related
unavoided energy-level crossings).

We shall be interested in another, qualitatively different
quantum-catastrophic scenario. It is encountered in the ${\cal
PT}-$symmetric quantum systems in which an isolated and real
energy level $E_n(\lambda)$ of a stable bound state grows with
$\lambda$ and touches the lower boundary of the essential
spectrum. Such a bound state may be perceived as disappearing
(or, in an opposite direction, emerging) at the
continuous-spectrum edge. 

The mathematically rather
sophisticated scenario of such a type is not too frequently
studied in the literature (cf. its samples in
\cite{Viola,Jonesdva}).
The gap will be partially filled in what follows. We shall consider
a multiparametric family of ${\cal PT}-$symmetric Hamiltonians $H$
and we shall search for a constructive guarantee of the coexistence
of bound states with the scattering solutions.

In section \ref{aa} we will introduce an ordinary difference
Schr\"{o}dinger equation which will prove useful for the purpose. In
subsequent sections \ref{aac} and  \ref{bb} we shall illustrate the
basic technical merits of our replacement of the real coordinates in
one dimension ($x \in \mathbb{R}$) by the discrete set of the grid
points ($x_{new} \in \mathbb{Z}$). We shall show that in a way
deviating from the conventional wisdom our preference of the
difference Schr\"{o}dinger equation is truly well motivated by
certain simplifications of the relevant parts of the underlying
mathematics. In the context of physics, our present choice of the
discrete interaction operators $V$ finds its important motivation in
the context of scattering. According to Refs.~\cite{Jones,discrete}
the very natural requirement of the causality and unitarity of the
$S-$matrix seems to {\em require} the use of the  weakly
non-local
interaction potentials of the present type.

Using the words of paper \cite{I} we may formulate our present task
as a ``study [of] the effect of the appearance of isolated
eigenvalues at the boundaries of the gaps in the essential
spectrum''. In section \ref{cc}, therefore, we shall turn attention
to the construction of the spectra. The weakly
non-local nature of our interaction potential $V$ will be shown to
lead to a number of answers to the phenomenologically inspired
qualitative questions.

In spite of an unusually flexible, $2J-$parametric
tridiagonal-matrix form of our interactions $V$, the model may be
claimed exactly solvable. Such a characteristic has two roots.
Firstly, it reflects the facilitated nature of the matching of
solutions in the models based on tridiagonal matrices \cite{machi}.
Secondly, in a way described in section \ref{dd} a decisive and
specific advantage of our model will be found in the possibility of
conversion of the traditional polynomial-zero-search bound-state
constructions into their closed-formula continued-fraction
alternatives, user-friendly and amenable to a straightforward
geometric interpretation.

Our results will be briefly discussed and summarized in
sections~\ref{ee} and \ref{summary}.

\section{The model\label{aa}}

\subsection{The real and discrete version of ${\cal PT}$ symmetry}

The doubly infinite discrete Laplacean
 \be
 H_0=
 \left [\begin {array}{rrr|rrr}
  \ddots&\ddots&&&&\\
  \ddots&2&-1&& &\\
 &-1&2&-1&&\\
 \hline
 {}&&-1&2&-1&\\
 && &-1&2&\ddots\\
 &&&&\ddots&\ddots
 \end {array}\right ]
 \label{kine}
 \ee
can be interpreted as a quantum Hamiltonian of a free particle
moving along the discrete real line, with the usual continuous
coordinate $x \in (-\infty,\infty)$ replaced by the grid of points
$x_n=n$ with, say, $ n = \ldots, -1, 0$ to the left and  $ n = 1, 2,
\ldots\, $ to the right from the partitioning line. As long as
matrix (\ref{kine}) is real and symmetric, it may be treated as
time-reversal symmetric. Formally, we write $H_0{\cal T}={\cal
T}H_0$. We may require that the time-reversal operator ${\cal T}$
performs the transposition (and, in general, also complex
conjugation) of matrices \cite{ali}.

In the spirit of review \cite{Carl}, the same Hamiltonian may be
also perceived as parity-time-reversal symmetric, $H_0{\cal
PT}={\cal PT}H_0$. The real and symmetric matrix of parity may be
chosen in the following antidiagonal matrix form
 \be
 {\cal P}=
 \left [\begin {array}{rlcrr}
  &&&&\ \ \dot{ \dot{\dot{}\   }\  }\\{}
  &&&\  1 &
 \\{}
 && 1 &&\\{}&1 \ &&&
 \\{}\ \ \ \
\dot{ \dot{\dot{}\   } \ }
 &&&&\end {array}\right ]\,
 \label{PTdef}
 \ee
of a square root of the unit matrix.

In the increasingly popular $ {\cal PT}-$symmetric quantum mechanics
\cite{book,Carl}, the discrete-matrix nature of the kinetic-energy
operator (\ref{kine}) and of its non-Hermitian $ {\cal
PT}-$symmetric perturbations
 \be
 H = H_0+V \neq H^\dagger
 \,,
 \ \ \ \ \
 H{\cal PT}={\cal PT}H
 \,
 \label{syst}
 \ee
played the key role, e.g., in the conceptual studies of the
scattering \cite{discrete,scatt,sigmaall}. These considerations were
facilitated by the $ {\cal PT}-$symmetric toy-model choices of
various special cases of the general real-matrix nearest-neighbor
interaction
 \be
 V=
 \left [\begin {array}{rrr|rrr}
  \ddots&\ddots&&&&\\
  \ddots&0&b&& &\\
 &-b'&0&a&&\\
 \hline
 {}&&-a'&0&b&\\
 && &-b'&0&\ddots\\
 &&&&\ddots&\ddots
 \end {array}\right ]\,.
 \label{potecy}
 \ee
In the light of the tridiagonal-matrix form of the kinetic-energy
component (\ref{kine}) of Hamiltonians (\ref{syst}) the same,
slightly non-local tridiagonality of the interaction could enrich
the physics without making the calculations too complicated. The
readers should also consult paper \cite{recur} in which we
demonstrated, that, and why, the choice of the non-Hermitian {\em
tridiagonal\,} real-matrix Hamiltonians leaves several technical
aspects of the constructive model-building unexpectedly
user-friendly and preferable.

In our present paper we shall study Hamiltonians (\ref{syst}) +
(\ref{potecy}), i.e., the doubly infinite matrices
 \be
 H=
 \left [\begin {array}{ccc|ccc}
  \ddots&\ddots&&&&\\
  \ddots&2&-1+b&& &\\
 &-1-b'&2&-1+a&&\\
 \hline
 {}&&-1-a'&2&-1+b&\\
 && &-1-b'&2&\ddots\\
 &&&&\ddots&\ddots
 \end {array}\right ]
 \,
 \label{douinf}
 \ee
in the bound-state dynamical regime. We shall assume that the
parameters form a real $2J-$plet $a,a',b,b',\ldots,z,z'$, i.e., that
the dynamics is controlled by the $2J$ independently variable real
matrix elements of $V$, with the outermost
$V_{-J,-J+1}=V_{J-2,J-1}=z$ and $V_{1-J,-J}=V_{J-1,J-2}=-z'$, etc.

The choice of Eq.~(\ref{douinf}) could be generalized to the $J \to
\infty$ limit with infinitely many free parameters (in this respect
see also section \ref{dd} below). Alternatively, the analogous $
{\cal PT}-$symmetric interaction Hamiltonians
 \ben
 H=
 \left [\begin {array}{ccc|c|ccc}
  \ddots&\ddots&&&&&\\
  \ddots&2&-1&& &&\\
 &-1&2&-1&&&\\
 \hline
 {}&&-1&2&-1&&\\
 \hline
 && &-1&2&-1&\\
 &&& &-1&2&\ddots\\
 & &&&&\ddots&\ddots
 \end {array}\right ]
 +
 \left [\begin {array}{ccc|c|ccc}
  \ddots&\ddots&&&&&\\
  \ddots&0&b&& &&\\
 &-b'&0&a&&&\\
 \hline
 {}&&-a'&0&a&&\\
 \hline
 && &-a'&0&b&\\
 &&& &-b'&0&\ddots\\
 &&&&&\ddots&\ddots
 \end {array}\right ]\,
 \een
of the odd dimensions after truncation and of a different
partitioning could have been considered as well. For the sake of
brevity, nevertheless, the latter version of the doubly infinite
Hamiltonians will not be studied in our present paper.

\subsection{Boundary conditions}

As long as $J < \infty$ remains finite, our present Schr\"{o}dinger
equation
 \be
 H |\psi\kt = E\,|\psi\kt\,
 \label{12r}
 \ee
is a linear difference equation of the second order, with the two
independent asymptotic solutions of the form
 \be
 |\psi\kt_n=const \times \exp(n\alpha)\,,
 \ \ \ \ \ |n| \gg J \geq 1
 \label{genesolu}
 \ee
where $\alpha \in \mathbb{C}$ is a suitable complex number related
to the (in general, complex) energy by the formula $2-E=2\, {\rm
cosh\,} \alpha$. These expressions represent the formally exact
solutions of Eq.~(\ref{12r}) in the whole free-motion kinematical
range, i.e., for $n \leq J$ and/or for $n \geq J-1$. Suitable
boundary conditions must be added.

\subsubsection{Scattering states}

The existence of the asymptotically traveling waves is possible
under condition ${\rm Re} \ \alpha=0$, i.e., for the energies which
belong to the continuous part of the spectrum of our Hamiltonians,
i.e., for  $E \in (0,4)$. In this dynamical regime a sample of the
non-numerical construction of the experimentally relevant reflection
and transmission coefficients may be found sampled, e.g., in
Ref.~\cite{discrete}.

\subsubsection{Bound states}

We shall exclusively pay attention to the bound state solutions of
our Schr\"{o}dinger Eq.~(\ref{12r}) and, in particular, to the
explicit constructive proofs of their existence. Once one requires
that these solutions decrease at large $|n|\gg J$, we have to split
the discussion according to the sign of our discrete coordinate $n$.

Along the negative discrete half-axis of $n$ we only have to admit
the exponentials (\ref{genesolu}) with $\alpha=\varphi$ such that
${\rm Re }\ \varphi> 0$. Having $2-E=2\, {\rm cosh\,} \varphi\ $
this means that whenever the bound-state energy itself becomes real,
it must be negative. In other words, our model does not admit the
stable, unitarily evolving bound states embedded in the continuum.

Along the positive discrete half-axis of $n$ we can only use
exponentials (\ref{genesolu}) with an opposite sign in
$\alpha=-\varphi$. This enables us to postulate
 \be
 |\psi\kt_n=\lambda  \times \exp[(n+J)\varphi]\,, \ \ n \leq -J
 \ee
and
 \be
  |\psi\kt_n=\varrho \times \exp[(J-1-n)\varphi]\,, \ \ n \geq
  J-1\,.
  \ee
These relations guarantee the normalizability of the bound states,
with the energy-representing parameter $\varphi$ still being complex
in general, restricted only by the above-mentioned bound-state
constraint of normalizability ${\rm Re }\ \varphi> 0$. A matching of
the asymptotics to the remaining part of our Schr\"{o}dinger
equation must be performed.

\section{Bound state in two-parametric $V$ ($J=1$)\label{aac}}

\subsection{The matching of wave functions in the origin}

About the first nontrivial $J=1$ Schr\"{o}dinger equation
 \be
  \left [\begin {array}{ccc|ccc}
  \ddots&\ddots&&&&\\
  \ddots&2-E&-1&& &\\
 &-1&2-E&-1+a&&\\
 \hline
 {}&&-1-a'&2-E&-1&\\
 && &-1&2-E&\ddots\\
 &&&&\ddots&\ddots
 \end {array}\right ]\,
 \left [
 \ba
 \vdots\\
 |\psi\kt_{-2}\\
 |\psi\kt_{-1}\\
 \hline
 |\psi\kt_0\\
 |\psi\kt_1\\
 \vdots
 \ea
 \right ]=0\,
 \label{hles}
 \ee
we may say that up to the respective constant multiplier the
bound-state eigenvectors are already known from the asymptotic
analysis,
 \be
 \ldots,
 \ \ \ \ |\psi\kt_{-2}=\lambda  \, \exp(-\varphi)\,,
 \ \ \ \ \
 |\psi\kt_{-1}=\lambda  \,,
 \ \ \ \ \
 |\psi\kt_0=\varrho \,,
 \ \ \ \ \
 |\psi\kt_1=\varrho \, \exp(-\varphi)\,,
 \ \ \ \ldots\,.
 \ee
This means that their matching will be mediated by the two innermost
rows of system~(\ref{hles}),
 \ben
 \left [\begin {array}{cc|cc}
 -1&2\, {\rm cosh\,} \varphi&-1+a&0\\
 \hline
 0 &-1-a'&2\, {\rm cosh\,} \varphi&-1
 \end {array}\right ]\,
 \left [
 \ba
 |\psi\kt_{-2}\\
 |\psi\kt_{-1}\\
 \hline
 |\psi\kt_0\\
 |\psi\kt_1
 \ea
 \right ]=0\,
 \een
i.e.,
 \ben
 \left [\begin {array}{cc|cc}
 -1&2\, {\rm cosh\,} \varphi&-1+a&0\\
 \hline
 0&-1-a'&2\, {\rm cosh\,} \varphi&-1
 \end {array}\right ]\,
 \left [
 \ba
 \lambda  \, \exp(-\varphi)\,\\
 \lambda  \,\\
 \hline
 \varrho \,\\
 \varrho \, \exp(-\varphi)\,
 \ea
 \right ]=0\,
 \een
i.e.,
 \ben
 \left [\begin {array}{cc}
   {\rm exp\,} \varphi&-1+a\\
 -1-a'&{\rm exp\,} \varphi
 \end {array}\right ]\,
 \left [
 \ba
    \lambda  \\
  \varrho
 \ea
 \right ]=0\,.
 \een
This relation has a nontrivial solution iff the determinant
vanishes,
 \be
  {\rm exp}\,2\, \varphi-(-1+a)(-1-a')=0\,.
  \label{trigi}
  \ee

\subsection{Bound-state-supporting
parameters}

%
%
\begin{figure}[h]                     
\begin{center}                         
\epsfig{file=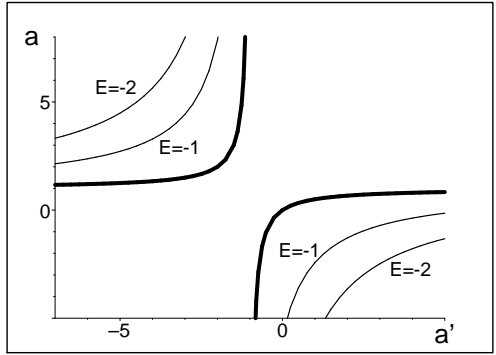,angle=270,width=0.4\textwidth}
\end{center}                         
\vspace{-2mm} \caption{The thick-curve boundary of the
doubly-connected domain ${\cal D}(a,a')$ supporting the bound state
in the $a-a'$ plane at $J=1$. The fixed-energy submanifolds are
sampled by the thinner hyperbolic curves.
 \label{fiut}}
\end{figure}

 \noindent
The analysis of solutions of the $J=1$ toy-model Eq.~(\ref{trigi})
is entirely elementary. It implies that in the two-dimensional plane
of parameters $a$ and $a'$ the boundary $\partial {\cal D}(a,a')$ of
the doubly connected domain ${\cal D}(a,a')$ of the existence of a
bound state will coincide with the hyperbolic curve $a=a'/(1+a')$
(cf. Fig.~\ref{fiut}). As long as the free parameters $a$ and $a'$
are real, we may abbreviate $\ (1-a)(1+a')=1+u(a,a')\ $ and have the
following obvious result.

\begin{lemma}

At $J=1$ the bound state exists iff
 $\ u(a,a')>0$.

\end{lemma}

 \noindent
In the $a-a'$ plane only the two lines with $a=1$ or $a'=-1$ are
singular. After their exclusion the bound state is found to exist
either for $a'>-1$ and $a<a'/(1+a')$ (i.e., under a hyperbolic-curve
maximum of $a$) or for $a'<-1$ and $a>a'/(1+a')$ (i.e., above a
hyperbolic-curve minimum of $a$). In the definition of the
Hamiltonian it would make sense to simplify the hyperbolic-curve
geometry of the critical-parameter boundaries in the $a-a'$ plane
and to replace the $a-a'$ plane, say, by the $u-a'$ plane (in which
the energies would not not depend on $a'$) or by the $a-u$ plane (in
which the energies would not depend on $a$). We shall see below that
the same type of reparametrization may be also recommended at any
higher number of parameters $2J>2$.

\begin{cor}
There are no $J=1$ bound states at $a=a'$.

\end{cor}

 \noindent
The latter observation sounds elementary but it could be extended to
all $J$, forming an important addendum to the minimally nonlocal
non-Hermitian models of Ref.~\cite{scatt} with $a=a'$, $b=b'$, etc.
These models were shown to support the scattering which may be made,
not quite expectedly \cite{Jones}, causal and unitary in an {\it ad
hoc\,} physical Hilbert space (cf. also
Refs.~\cite{discrete,scatt,SIGMAtri} in this respect).

\section{Matching conditions at $J \geq 2$\label{bb}}

\subsection{The four-parametric $V$ as a methodical guide ($J=2$)}

At $J=2$ the insertion of the available exact solutions
 \be
 \ldots,
 \ \ \ \ |\psi\kt_{-3}=\lambda  \, \exp(-\varphi)\,,
 \ \ \ \ \
 |\psi\kt_{-2}=\lambda  \,,
 \ \ \ \ \
 |\psi\kt_1=\varrho \,,
 \ \ \ \ \
 |\psi\kt_2=\varrho \, \exp(-\varphi)\,,
 \ \ \ \ldots\,.
 \ee
reduces the complete Schr\"{o}dinger equation to the four matching
conditions
 \ben
  \left [\begin {array}{ccc|ccc}
  -1&2\, {\rm cosh\,} \varphi&-1+b&0& 0&0\\
 0&-1-b'&2\, {\rm cosh\,} \varphi&-1+a&0&0\\
 \hline
 {0}&0&-1-a'&2\, {\rm cosh\,} \varphi&-1+b&0\\
 0&0&0 &-1-b'&2\, {\rm cosh\,} \varphi&-1
 \end {array}\right ]\,\,
 \left [
 \ba
 \lambda  \, \exp(-\varphi)\,\\
 \lambda \\
 |\psi\kt_{-1}\\
 \hline
 |\psi\kt_0\\
 \varrho \\
 \varrho \, \exp(-\varphi)
 \ea
 \right ]=0\,
 \een
{\it alias}
 \ben
  \left [\begin {array}{cc|cc}
  {\rm exp\,} \varphi&-1+b&0& 0\\
 -1-b'&2\, {\rm cosh\,} \varphi&-1+a&0\\
 \hline
 0&-1-a'&2\, {\rm cosh\,} \varphi&-1+b\\
 0&0 &-1-b'& {\rm exp\,} \varphi
 \end {array}\right ]\,\,
 \left [
 \ba
 \lambda \\
 |\psi\kt_{-1}\\
 \hline
 |\psi\kt_0\\
 \varrho
 \ea
 \right ]=0\,.
 \een
Once we abbreviate
 $$(1-a)\,(1+a')= 1+u\,,\ \ \ \ (1-b)\,(1+b')=1+v$$
and once we eliminate $a=1-(1+u)/(1+a')$ and $b=1-(1+v)/(1+b')$ and
require that
  \be
  \exp 2\,\varphi=x^2=t > 1\,,
  \label{varia}
  \ee
an elementary algebra leads to the amazingly compact and exactly
solvable secular equation
 \be
 t^2-t\,(1+u+2\,v) +v^2=0
 \label{13t}
 \ee
{\it alias}
 \be
 (t-v)^2=t(1+u)
 \label{14t}
 \ee
{\it alias}
 \be
 (t-v)=\pm \sqrt{t(1+u)}\,.
 \label{15t}
 \ee
With $1+u>0$ and for  $v>v_{minimal}=-(1+u)/4$ we see that in
Eq.~(\ref{15t}) the left-hand-side straight line {\em always}
intersects the right-hand-side parabola at the two real values of
$t>0$. From Eq.~(\ref{14t}) we may deduce, equivalently, that the
right-hand-side straight line {\em always} intersects the
left-hand-side parabola at the same two real values of $t>0$.
Naturally, only the roots with the property $t>1$ correspond to the
physical bound states.

\begin{figure}[h]                     
\begin{center}                         
\epsfig{file=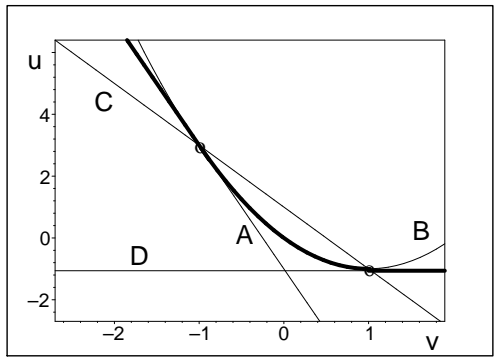,angle=270,width=0.4\textwidth}
\end{center}                         
\vspace{-2mm} \caption{The thick-curve boundary of the domain
supporting the $J=2$ ground state in the $u-v$ plane. The energies
become complex for $u<-1-4v$, i.e., to the left from the line ``A''.
This is the domain's boundary at $v < -1$. The parabola of
$u=(v-1)^2-1$ (curve ``B'') forms the boundary for $v \in (-1,1)$
while the line ``C'' with $u=1-2v$ is never the boundary because at
the larger $v > 1$ the role of the boundary is played by the second
energy-reality restriction $u>-1$, i.e., by the horizontal line
``D''.
 \label{fihr}}
\end{figure}

\begin{lemma}
\label{leje}

At $J=2$ the ground state exists in the upper part of the $u-v$
plane, viz., in the domain of $u>u_{min}(v)$ bounded, from below, by
the thick curve $u_{min}(v)$ of Fig.~\ref{fihr}.

\end{lemma}

\begin{proof}
The larger root
 $$t_+ =\ha \left [1+u+2v+ \sqrt{(1+u+4v)(1+u)}\right ]$$
of the quadratic algebraic Eq.~(\ref{13t}) remains real to the right
from the line ``A'' of Fig.~\ref{fihr} (which is decisive for $v <
-1$) and above the line ``D'' (which is decisive for $v >1$). The
parabolic segment of the boundary represents the restriction
$t_+>1$.

\end{proof}

\begin{lemma}

\label{ledv}

The $J=2$ model can also support the first excited state. The
two-dimensional parametric subdomain of its existence is doubly
connected. As a subdomain of the domain of Lemma~\ref{leje} it
shares its lower bound but lies {below} the parabolic curve ``B'' of
Fig.~\ref{fihr}.

\end{lemma}

\begin{proof}
The smaller root
 $$t_+ =\ha \left [1+u+2v- \sqrt{(1+u+4v)(1+u)}\right ]$$
can only become bigger than one when its square-root part is
sufficiently small. The opposite side of the parabolic boundary
becomes allowed, therefore.

\end{proof}

 \noindent
The first excited state can
only exist at $v \notin (-1,1)$; either for $u \in
(-1,u_{max}^{(+)}(v))$ at positive $v>1$, or for positive $u>3$ in
an interval of $u \in (u_{min}^{(-)}(v),u_{max}^{(-)}(v))$ at
negative $v<-1$.
These observations may be checked in Fig.~\ref{fihr}.

\subsection{General effective secular determinants (any $J\geq 2$)}

The general $J\geq 2$ matching condition has the $2J-$dimensional
matrix form
 \be
 \label{maco}
  \left [\begin {array}{ccccccc}
  {\rm exp\,} \varphi&-1+z&0& \ldots& &\ldots& 0\\
 -1-z'&2\, {\rm cosh\,} \varphi&-1+y&\ddots& && \vdots\\
 0&\ddots&\ddots&\ddots&0&& \\
  \vdots&\ddots&-1-a'&2\, {\rm cosh\,} \varphi&-1+b&\ddots&\vdots\\
 & &0&\ddots&\ddots&\ddots&0\\
 \vdots&& &\ddots&-1-y'&2\, {\rm cosh\,} \varphi&-1+z\\
 0&\ldots& &\ldots&0 &-1-z'& {\rm exp\,} \varphi
 \end {array}\right ]\,\,
 \left [
 \ba
 \lambda \\
 |\psi\kt_{-J+1}\\
 \vdots\\
 |\psi\kt_{-1}\\
 |\psi\kt_0\\
 \vdots\\
 |\psi\kt_{J-2}\\
 \varrho
 \ea
 \right ]=0\,.
 \ee
This equation is, in effect, the Feshbach's \cite{Feshbach}
model-space reduction of our initial doubly infinite Schr\"{o}dinger
equation. In general, the energy-dependence of the Feshbach's
effective Hamiltonian would be complicated. In order to be able to
proceed, one has to return to the further few nontrivial
algebraic-manipulation experiments with the model.

\section{The six-parametric solutions ($J=3$)\label{cc}}

One may expect that the evaluation of the secular determinant of
Eq.~(\ref{maco}) and the search for its bound-state-energy zeros
might proceed along the same lines as in the previous example with
$J=2$, in principle at least. In practice, the prospect of the
non-numerical solution of Eq.~(\ref{maco}) seems less encouraging.
Even the very evaluation of the secular determinant of
Eq.~(\ref{maco}) would be hardly feasible without a standard
symbolic manipulation software.


\subsection{The change of variables}

At the first sight, even the use of a computer algebra does not seem
too promising. A preliminary naive test may be performed to show
that the length of the typical results (say, of a formula for the
secular determinant at any $J\geq 3$) already requires too many
printed pages for being of any use. One must return again to the
heuristic experiments, still at the not too large integers $J$.

The construction of the two-dimensional domains of the physical
parameters and of their boundaries (cf. Fig.~\ref{fihr}) as well as
the proofs of Lemmas \ref{leje} and \ref{ledv} at $J=2$ were too
elementary. They relied heavily upon the availability of the
explicit formulae for the bound-state energies. Such a naive
approach fails to work even at the next choice of $J=3$. As
expected, the evaluation of the $J=3$ secular determinant of
Eq.~(\ref{maco}) (containing seven variables) was already hardly
feasible by pen and pencil. Even the symbolic-manipulation software
did not help too much. We had to conclude that the use of the
computer merely transfers the difficulty from the feasibility of the
evaluation to the readability of the formula. The ``exact'' secular
polynomial seemed too long for offering any insight. We imagined
that before any transition to the higher $J \geq 3$ the naive
approach to the construction of the bound states via matching
(\ref{maco}) had to be thoroughly reanalyzed and amended.

At $J=3$ the first success in our methodical experiments came with
the $(J=2)-$inspired tentative change of parametrization
 \be
 \label{chapa}
 (1-a)\,(1+a')= 1+u\,,
\ \ \ \ (1-b)\,(1+b')=1+v\,\ \ \ \
(1-c)\,(1+c')=1+w
 \ee
followed by the elimination
 $$
 a=1-(1+u)/(1+a')\,,
\ \ \ \ b=1-(1+v)/(1+b')\,,
\ \ \ \ c=1-(1+w)/(1+c')\,.
 $$
The trick worked and shortened the formulae (say, for the secular
determinants) back to a tractable size. This encouraged us to
prolong the series of reparametrizations (\ref{chapa}),
 \ben
 (1-d)(1+d')= 1+y,
\ \  (1-e)(1+e')=1+z,\ \
(1-f)(1+f')=1+m,\ \
(1-g)(1+g')=1+n
 \een
etc (where, as the famous Littlewood's
joke says, ``$e$ need not be equal to 2.718\ldots'').

\subsection{The emergence of the ground state from continuum}

At $J=3$ the Feshbach's effective secular determinant remained
almost polynomial not only in the energy variable $x=\exp \varphi
> 1$ but also in its square $t=\exp 2\,\varphi > 1$. The determinant
appeared to be a sum of a polynomial and of a single, utterly
elementary inverse-power term $w^2/t$. As long as $t>1$, the latter
term may be considered nonsingular. After a pre-multiplication by
$t$ the secular equation acquires the quartic polynomial form
 \be
{t}^{4}- \left(1+u+2\,v+2\,w \right) {t}^{3}+ \left( 2\,uw+(v+w)^2
\right) {t}^{2}+ \left( {w}^{2}(1-u)+2\,vw \right) t +{w}^{2}=0\,.
 \label{secu3}
 \ee
A surprise emerged when we checked the behavior of this equation
near the origin of the three-dimensional space of the new, reduced
parameters $u$, $v$ and $w$.

\begin{lemma}
\label{leduj} For the sufficiently small, ${\cal O}(\lambda)$
parameters $u$, $v$ and $w$ there exists a unique non-small, ${\cal
O}(1)$ energy root
 \be
 t_0 = 1+u+2\,v+2\,w + {\cal O}(\lambda^2)\,.
 \label{maxsim}
 \ee
This means that inside the half-space with $u+2\,v+2\,w>0$ (i.e., in
particular, in the positive-parameter quadrant with $u>0$, $v>0$ and
$w>0$) there exists a physical bound state which emerges from the
continuum, with $E_0 = -(u+2\,v+2\,w)^2+ {\cal O}(\lambda^3)$.

\end{lemma}

\begin{proof}
Up to quadratically small corrections Eq.~(\ref{secu3}) reads $\
t^4=\left(1+u+2\,v+2\,w \right) {t}^{3}$. This implies that as many
as three of the roots remain small (i.e., unphysical, much smaller
than one). Thus, we get $\varphi_0 =u+2\,v+2\,w+{\cal
O}(\lambda^2)$.
 \end{proof}
 \noindent
In the light of our numerous symbolic-manipulation experiments the
same conclusion seems to remain valid at any number $J$ of the
``reduced'' parameters $u$, $v$, \ldots\,.

\begin{conj}
At any integer $J$ there exists a physical bound state which emerges
from the continuum in the positive-parameter vicinity of the origin
$u=0$, $v=0$, \ldots\,.
\end{conj}

%
\begin{figure}[h]                     
\begin{center}                         
\epsfig{file=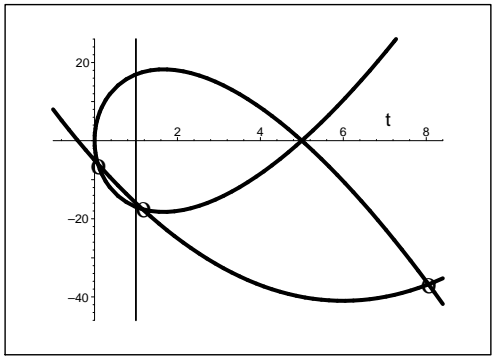,angle=270,width=0.4\textwidth}
\end{center}                         
\vspace{-2mm} \caption{Graphical proof, via Eq.~(\ref{graso}), of
the existence of the two excites states at $J=3$, $u=17$, $v=6$ and
$w=5$.
 \label{fifi}}
\end{figure}

\subsection{The emergence of the excited states from continuum}

The choice of $J=3$ leads to the quartic-polynomial secular
Eq.~(\ref{secu3}). Apparently, this is the ultimate special case of
our Hamiltonian (\ref{douinf}) from which one could extract the set
of the energy roots as functions of the couplings in closed form, in
principle at least. Naturally, it is necessary to add that the
concrete form of the available exact formulae for $E_n(u,v,w)$ is
very complicated and makes them practically useless. Fortunately,
our ``last exactly solvable'' secular Eq.~(\ref{secu3}) has a
special form which may be pre-factorized,
 \be
 t^2\,(t-v-w)^2-t\,(1+u)\,(t-w)^2-2\,w\,t\,(t-v-w)+w^2=0\,,
 \ \ \ \ \ J=3
 \label{sec3}
 \ee
and, in the next step,
 \be
 [t(t-v-w)-w]^2=t(1+u)(t-w)^2\,,
 \ \ \ \ \ J=3\,.
 \label{legraso}
 \ee
This enables us to replace it by its simpler, square-rooted
equivalent version
 \be
 [t(t-v-w)-w]=\pm (t-w)
 \sqrt{t(1+u)}\,.
 \label{graso}
 \ee
Such an equation is already much more transparent: its
right-hand-side may be visualized as a loop while the left-hand side
is a parabola. Thus, a graphical solution of this version of secular
equation becomes elementary. It may be found sampled in
Fig.~\ref{fifi} which represents, at the preselected triplet of
parameters, a graphical proof of the existence of the first two
excited bound states with negative energies, i.e., below the
continuum.

In the picture the left-hand-side parabola intersects the
right-hand-side loop at the three small circles (the fourth,
rightmost, ground-state intersection is out of the range of the
picture). The thinner vertical line marks the left boundary of the
interval of admissible $t\in (1,\infty)$. Thus, the leftmost
intersection is spurious, not representing a bound state. It would
yield a purely imaginary $\varphi$ and an unphysical wave function
solution, incompatible with the physical Dirichlet's asymptotic
boundary conditions.

\section{Multiparametric solutions ($J \geq 3$)\label{dd}}

\subsection{Sturmians as exact solutions}

The inspection of formulae  (\ref{15t}) (where $J=2$) and
(\ref{graso}) (where $J=3$) inspired us to make use of the following
ansatz,
 \be
  \sqrt{(1+u)t} =\pm f_J(t,v,w,\ldots)\,,
 \ \ \ \ \ J = 2, 3, \ldots\,.
 \label{closed}
 \ee
The hypothesis was based on the first three already available
formulae,
 \be
  f_1(t)=t\,,\ \ \ \
  f_2(t,v)= t-v\,,
  \ \ \ \
  f_3(t,v,w)=
    {\frac {t^2-(v+w)\,t-w }{t-w}}\,.
    \label{thede}
 \ee
The structure of these formulae can be extrapolated to all of the
higher values of $J$. We verified that the right-hand-side functions
have the form of the ratio of two polynomials,
 \be
  f_J(t,v,w,\ldots)=
    {\frac {t^{J-1} - c_1\,t^{J-2} + \ldots }{t^{J-2}
     - d_1\,t^{J-3}+\ldots}}\,.
    \label{therude}
 \ee
This observation led us to the conclusion that

\begin{itemize}

\item
it makes sense to keep using ansatz (\ref{closed}) + (\ref{therude})
and to treat the formula as the basic form of the solution of our
present bound-state problem.

\item
it makes sense to abandon the traditional recipe in which the
energies are sought as functions of the couplings. It seems optimal
to replace this strategy by the formally equivalent but technically
simpler approach in which the given-energy bound states are called
``Sturmians'' \cite{Sturmian}.

\end{itemize}

 \noindent
Once we decided to select the innermost coupling $u_n$ as a quantity
to be sought for, we may treat it as a function $u_n(t,v,w,\ldots)$
of the remaining couplings and of the {\em freely variable} energy
parameter $t \in (1,\infty)$. The resulting implicit-function
picture of the quantum system will have the advantage of
circumventing the complicated search for the energy roots $t_n$. In
other words, the hypothesis of the Sturmian bound-state solvability
provides an explicit definition of the bound states at all $J$.

One of the main practical consequences of such an innovated approach
may be seen in the reduction of the solution of the matching
condition (\ref{closed}) to the search for intersections between the
right-had-side curve $f_J(t,v,w,\ldots)$ with the left-hand-side
parabola $\pm \sqrt{(1+u)t}$. In this presentation the width of the
parabola is exclusively controlled by the innermost coupling $u>-1$
while the factorized rational function
 \be
  f_J(t,v,w,\ldots)=
 \frac{\prod_{k=1}^{J-1}(t-\beta_k)}{\prod_{m=1}^{J-2}(t-\gamma_m)}\,
    \label{gether}
 \ee
has the growing asympototes $ f_J(t,v,w,\ldots)=t+{\cal O}(1)$ and
(in general, complex) poles at the values of $t=\gamma_m$ with $m =
1, 2, \ldots, J-2$.

In the most elementary scenario all of the latter positions of the
singularities $\gamma_m$ and zeros $\beta_k$ may be assumed real.
The global shape of the function $f_J(t,\ldots)$ will be then
determined by the distribution of the $J-1$ zeros of the numerator
among the $J-1$ intervals $(-\infty,\gamma_1)$,
$(\gamma_1,\gamma_2)$, \ldots, $(\gamma_{J-2},\infty)$.

The menu of the possible intersection patterns with the
$u-$dependent parabola may be then illustrated at $J=3$, with just
two intervals $(-\infty,\gamma_1)$ (=``left'', $L$) and
$(\gamma_{1},\infty)$ (=``right'', $R$). Then, with both of the
zeros of the numerator in $L$, we will receive a hyperbolic curve of
$f_3(t)$ composed of its $\bigcap -$shaped branch to the left and of
its $\bigcup -$shaped branch to the right. Next, for both of the
zeros of the numerator in $R$, the shape of the $v-$ and
$w-$dependent hyperbola $f_3(t)$ (and, hence, also the discussion of
its possible real intersections with the $u-$dependent parabola)
will remain qualitatively the same.

A different intersection pattern will only be encountered in the
third possible case with $\beta_1\in L$ and $\beta_2\in R$. In this
case the hyperbolic curve $f_3(t)$ with its two asymptotes $f_3(t)
\sim t$ and $f_3(t)\sim -1/(t-\gamma_1)$ will represent a map of
either $L$ or $R$ on the whole real axis so that one will always
have to deal with the four real intersections. Thus, the
classification of the whole bound-state spectrum will be restricted
to the identification of the ``physical'' intersections such that
$t>1$ (cf.~Eq.~(\ref{varia})).

At the larger counts of parameters $J>3$ the discussion will remain
qualitatively the same. In every interval of $t \in
(\gamma_{m},\gamma_{m+1})$ one merely has to distinguish between the
occurrence of the {\em odd\,} number of the zeros $\beta_k$ (leading
to the asymptotically $\bigcap -$shaped or $\bigcup -$shaped forms
of $f_J(t)$) and the occurrence of the {\em even\,} number of the
zeros $\beta_k$ leading to the map of the interval on the whole real
axis, i.e., to the guarantee of the existence of at least two real
intersections with the $u-$dependent parabola in the interval.

\subsection{Partial fraction expansions}

The main encouragement of our present study of multiparametric
interactions (\ref{potecy}) originated from the $J=3$ formula
(\ref{graso}) which became amazingly elementary in its
partial-fraction-expanded version
 \be
 f_3(t,v,w)=
    t-v-{\frac {(1+v)w}{t-w}}\,.
    \label{thedeer}
 \ee
One of the immediate consequences of this expansion is the
simplified classification of the shapes of $f_3(t,v,w)$. This
function has the form of the  $\bigcap$ + $\bigcup$-shaped hyperbola
iff $(1+v)w<0$. {\it Vice versa}, the $/$ + $/$-shaped
two-growing-curves hyperbola (i.e., a guarantee of the existence of
the four intersections with the $u-$dependent parabola) takes place
at $(1+v)w>0$ (naturally, the spike disappears at  $(1+v)w=0$).




The latter observations combine the algebraic reduction of formulae
with an enhancement of an intuitive insight and graphical
classifications of the curve-intersections and of the related
alternative dynamical scenarios. This is a decisive final
simplification of the discussion which could be also achieved at the
higher parameter counts $J>3$. The  test has been performed in the
eight-parametric model. Using the computer-assisted symbolic
manipulations we revealed that near $t=0$ the secular determinant
exhibits the singularity $y^2/t^2+{\cal O}(t^{-1})$ (let us remind
the readers that $y=(1-d)\,(1+d')-1$). After its subsequent
pre-multiplication by $t^2$ it is replaced by its regular version
 $$
{t}^{6}+ \left( -1-2\,w-2\,y-u-2\,v \right) {t}^{5}+ \left(
2\,wy+4\,v y+2\,uw+{v}^{2}+2\,uy+{w}^{2}+2\,vw+{y}^{2} \right)
{t}^{4}+
 $$
 $$
 +\left( -{
y}^{2}u+{y}^{2}-2\,wvy-{w}^{2}u-2\,wuy+2\,uy+2\,wy-2\,{y}^{2}v-2\,{v}^
{2}y+2\,vy+2\,vw+{w}^{2} \right) {t}^{3}+
 $$
 $$
 +\left(
2\,wy+{y}^{2}+{y}^{2}
{v}^{2}-2\,wvy-2\,wuy-2\,{y}^{2}u+{w}^{2}+2\,vy-2\,{y}^{2}v \right)
{t }^{2}+
 $$
 $$
 +\left( -2\,{y}^{2}v+2\,wy-{y}^{2}u+{y}^{2} \right)
t+{y}^{2}=0\,.
 $$
The factorization of this long equation for the $u-$coupling
Sturmians is again found by the computer and leads to the final
result of the form (\ref{closed}),
 \be
 \left( -y-ty+tyv-{t}^{2}y-{t}^{2}v+{t}^{3}-tw-{t}^{2}w \right)
 =\pm \left(
 t^2-(w+y)\,t-y
  \right) \,\sqrt{t(1+u)}\,.
 \label{hugraso}
 \ee
This is to be re-read as the definition of $f_4(t,v,w,y)$ in
formula~(\ref{closed}), i.e., as the definition of the Sturmian
bound-state coupling $u_n=u_n(t,v,w,y)$, i.e., as our ultimate
closed-form description of the bound states.

\subsection{Continued fractions}


The evaluation of the partial-fraction expansion of $f_4(t,v,w,y)$
was only rendered possible by the use of computer software but the
computer provided the  fairly elementary formula rather quickly,
 $$
 f_4(t,v,w,y)=t-v-{\frac {y+tw+yv+tvw}{{t}^{2}+ \left( -y-w \right)
 t-y}}\,.
 $$
We noticed that it seems quite natural to apply the same
partial-fraction expansion algorithm to the denominator. This led to
the truly elementary finite-continued-fraction formula
%
%
 $$
 f_4(t,v,w,y)= t-A_0- \cfrac{B_1}{t-A_1
  - \cfrac{\widetilde{B_2}}{t-\widetilde{A_2}} } \ ,
 $$
where $A_0=A_0(v)=v$ was already determined at $J=2$, where
$B_1=B_1(v,w)=(1+v)w$ was already determined at $J=3$ and where the
determination of $A_1=A_1(w,y)=w+y+y/w$ was new. The last two
coefficient functions
 $
 \widetilde{B_2}=-(1+w)y^2/w^2$
 and $\widetilde{A_2}=-y/w\,
 $
are marked by the tildas. The reason is that the determination of
these functions can only be completed on a higher level of
calculations using $J=5$ and $J=6$. The ultimate, untilded formulae
for
 \be
 B_2=B_2(w,y,z)=(1+w)(y^2-(1+y)wz)/w^2
 \label{bedva}
 \ee
and
 \be
 A_2=A_2(w,y,z,m)=
 -{\frac {y}{w}}+z+{\frac {yz-(1+y)(1+z)wm}{{y}^{2}-wz(1+y)}}
 \label{adva}
 \ee
have been evaluated using models with $J=5$ and $J=6$, therefore.

In a compatibility check one may consider the limit $y \to 0$ in the
new $J=4$ continued fraction confirming the expected degeneracy to
its $J=3$ predecessor. Similarly, the limit $z \to 0$ returns the
coefficient $B_2$ to its incomplete, reduced, tilded $J=4$ version.

By induction one arrives at the general continued fraction formula
 $$
 f_J(t,v,w,\ldots)=t-A_0- \cfrac{B_1}{t-A_1 - \cfrac{{B_2}}{t-{A_2}-
  \cfrac{{B_2}}{t-{A_3}-\ddots} } \ }\,.
 $$
which is valid at all $J$. Whenever the integer $J$ remains finite,
this continued fraction will terminate via the appearance of a
tilded, vanishing value of the coefficient $\widetilde{B_{J-1}}=0$.
In the opposite direction, the complete, untilded definition of the
coefficients $B_k$ and $A_k$ will, in general, result from the
routine evaluation and analysis of the secular determinants at
$J=2k+1$ and $J=2k+2$, respectively.

\section{Discussion\label{ee}}

\subsection{Quantum systems in the discrete-coordinate quasi-Hermitian
representations}

The study of quantum systems in their quasi-Hermitian and ${\cal
PT}-$symmetric {\em discrete\,} representations usually originates
from the needs of open-system studies \cite{Joglekar} or of
classical optics \cite{Makris}, especially in the light of the
current quick developments of nanotechnologies \cite{symmetry}. Here
we proceeded in a complementary direction of connecting these models
with the simulations of the various forms of quantum phase
transitions (cf., e.g., \cite{Graefe,Rusicka}).

In the latter setting people feel most often inspired by the
Bender's and Boettcher's \cite{BB} conjecture that beyond the scope
of the conventional textbooks, the unitary evolution of quantum
systems may still be described as generated by certain non-standard
Hamiltonians $H$ with real spectra. In the language of our recent
review \cite{SIGMA}, these operators $H$ are interpreted as
self-adjoint in the (by assumption, overcomplicated) ``standard''
physical Hilbert space ${\cal H}^{(S)}$ (where one might write
$H=H^\ddagger$) but, {\em at the same time}, non-self-adjoint in a
``friendlier'' manifestly unphysical and auxiliary Hilbert space
${\cal H}^{(F)}$ (where we have, using the most conventional
notation, $H \neq H^\dagger$).

There exist good reasons (cf., e.g., \cite{Jones,Trefethen} or
\cite{Siegl}) why one should be very careful when leaving the safe
mathematics of self-adjoint operators in $L^2(\mathbb{R})$ and when
extending the study of the Schr\"{o}dinger's bound-state equations
(with real spectra) to the models with the complex (e.g., ${\cal
PT}-$symmetric \cite{Carl} or, more generally, pseudo-Hermitian
\cite{ali}) local interactions $V=V(x) \neq V^*(x)$. One of these
reasons is that, by assumption, we have $H = T+V \neq H^\dagger$ in
$L^2(\mathbb{R})$. Hence, the latter (and, certainly, maximally
user-friendly) Hilbert space must be declared unphysical, i.e., in
the present notation, we have $L^2(\mathbb{R})\ \equiv \ {\cal
H}^{(F)}$. Then, the theory (cf., e.g., its compact formulation in
\cite{SIGMA}) tells us that the physical requirement $H=H^\ddagger$
of the self-adjointness in ${\cal H}^{(S)}$ finds its equivalent
representation
 \be
 H^\dagger \Theta = \Theta\,H\,,
 \ \ \ \ \ \ \Theta=\Omega^\dagger\Omega\,
 \label{natur}
 \ee
in ${\cal H}^{(F)}$, mediated by an introduction of an auxiliary
metric operator $\Theta \neq I$.

Under certain additional subtler mathematical conditions the
validity of the latter relation enables us return to the quantum
theory of textbooks while calling the Hamiltonian itself, for the
sake of clarity of the terminology, quasi-Hermitian \cite{Geyer}. In
\cite{SIGMA} the readers may also find an exhaustive explanation of
the role and origin of the so called Dyson's maps $\Omega$ and/or of
the physical Hilbert-space metric operator $\Theta$ entering the
formula. The mathematical meaning of this condition is, in
principle, elementary \cite{Dieudonne}. Its appeal in physics was
discovered by Dyson \cite{Dyson} and widely used in nuclear physics
\cite{Geyer}. In essence, the relation just reflects the
mathematical compatibility of a given non-Hermitian Hamiltonian $H$
(with real spectrum) with the natural physical requirement of its
observability \cite{ali,arara}.

%

In practice, the key problem is that for a given $H$, we must find
at least one operator $\Theta$ which would satisfy
Eq.~(\ref{natur}). In paper \cite{recur} we showed that such a
search remains exceptionally straightforward if and only if the
matrix representation of $H$ is real and tridiagonal. This was,
after all, one of the most important mathematical reasons for our
present choice of model (\ref{douinf}).

One should add that in the applied quantum mechanics one often
encounters a conflict between the tractability and flexibility of
the available phenomenological models. The exactly solvable ones are
usually not too flexible, and {\it vice versa}. This may be also
perceived as a supportive argument in favor of the discrete models.
After the acceptance of new physics behind ${\cal PT}-$symmetric
potentials \cite{Carl} a widely welcome progress has recently been
achieved in this field \cite{lattices}. A perceivable extension of
the menu of the eligible tractable Hamiltonians has been obtained.
The common textbook self-adjoint versions of quantum models were
complemented by a new class using non-Hermitian Hamiltonians.

As long as the precise range of the applicability of the theory
using Hamiltonians $H =T+V \neq H^\dagger$ is not yet fully
understood, the use of the discrete models has an advantage of not
being directly exposed to the Trefethen's \cite{Trefethen} purely
mathematical criticism (based mainly on the work with differential
operators, cf. also \cite{Siegl}). Moreover, the discretization very
well fits the otherwise rather counterintuitive bounded-operator
preference as recommended, in reaction to the Dieudonn'e's
\cite{Dieudonne} slightly more abstract criticism, by nuclear
physicists \cite{Geyer}.

\subsection{The physical appeal of quantum systems exhibiting
phase transitions}

The birth of quantum mechanics might have been dated by the
resolution of the classical-physics paradox of the experimentally
observed stability of the large number of the atoms and molecules.
In the language of mathematics the explanation was provided by a
sophisticated identification of the measured discrete quantities
with the elements of the spectrum of the corresponding self-adjoint
operator. Thus, typically, the predicted real values of the energy
levels $E_n$, say, in hydrogen atom were shown obtainable, from a
suitable self-adjoint Hamiltonian $\mathfrak{h}
=\mathfrak{h}^\dagger$, as its eigenvalues. In the methodical
setting it has been found useful to illustrate the theory via the
one-dimensional Schr\"{o}dinger equations with various real and
local (and, often, exactly solvable \cite{Cooper}) confining
potentials $V=V(x)$.

One of the  weak points of the textbook theory may be seen in the
necessity of an appropriate assignment of the self-adjoint
Hamiltonian $\mathfrak{h}$ to a quantum system in question. From a
purely pragmatic perspective this assignment has two deep aspects.
Firstly, in the context of phenomenology, the self-adjointness of
$\mathfrak{h}$ is a robust property which does not offer a suitable
tool for the description of any finite-life or resonant quantum
systems which often occur in the nature and for which the energy
levels should not be strictly real. Secondly, in opposite direction,
the stability (i.e., the reality of the energy levels) may be
guaranteed even when the Hamiltonian itself is not a self-adjoint
operator. A number of the exactly solvable illustrative examples of
the one-dimensional differential Schr\"{o}dinger equation may be
found, e.g., in Ref.~\cite{sgezou}.

The later option has been made popular by Bender and Boettcher
\cite{BB} and it extends the practical applicability of quantum
theory. In its final formulation the option led to an upgrade of
quantum mechanics (cf., e.g., its reviews \cite{Carl,ali}). In
essence, the extended family of many new manifestly non-selfadjoint
Hamiltonians (with real spectra) was found compatible with the
standard quantum mechanics of textbooks (cf. \cite{Geyer}; a number
of the related mathematical as well as historical comments may be
also found in the most recent book \cite{book,MZbook}).

In the literature, unfortunately, the opinions concerning the
nontrivial applications of the new theory are not yet unified.
Nevertheless, in all of the existing versions of the theory the
existence of the nontrivial real boundaries $\partial {\cal D}$ of
applicability belong to the cornerstones of the theory. Their
localization and toy- model descriptions belong to the most
interesting challenges and subjects of the future research.
Certainly their study is openings new ways towards our understanding
of the phenomena of the loss of the quantum stability or
observability.

\subsection{Computer-capacity limitations of the constructions}

In the case of our present ${\cal PT}-$symmetric toy model
(\ref{douinf}) the applicability of the symbolic manipulations
yielding exact solutions encounters the two main limitations in
practice. The first one lies in the quick growth of the number of
the independently variable parameters. This is an obstacle which
already made some of the formulae next to useless (i.e., too long
for being displayed in print) as early as at $J=3$.

The second difficulty concerns the hidden nature of the information
carried by the long formulae. The picture only remains fully
transparent at $J=1$. In this case the physical parametric domain
${\cal D}(a,a' )$ has been found to have a clear and elementary
shape (cf.~Fig.~\ref{fiut}). Even the very next, $J=2$ physical
parametric domain ${\cal D}(a,a',b,b')$ already becomes
four-dimensional, i.e., not too easily displayed.

One of the specific merits of our present class of models is that
the latter difficulties may be softened. For example, after the
discovery of the energy-preserving submanifolds (sampled by the two
thin hyperbolic curves in Fig.~\ref{fiut}), we were able to halve
the $2J-$plet of the physical parameters to the mere $J-$plet. In
Fig.~\ref{fihr} we displayed the reduced, two-dimensional physical
domain ${\cal D}(u,v)$ also at $J=2$.

Thanks to several further unexpected simplifications of the
representation of the bound states we were able to extend the
transparent graphical interpretation of the spectra to the
six-parametric model with $J=3$ (cf. its sample in Fig.~\ref{fifi})
and, in principle, also to the eight-parametric model with $J=4$
(cf.~Eq.~(\ref{hugraso})).

A quick increase of the time and memory requirements limited our
present constructions of the Sturmians to not too large $J>4$. Using
a small PC we still managed to work with $J=5$ and $J=6$. We
obtained the above-displayed closed and still comparatively compact
formulae (\ref{bedva}) and (\ref{adva}) for the entirely general
continued-fraction coefficient functions $B_2$ and $A_2$.

Unfortunately, the next round of calculations with $J=7$ and $J=8$
already reached, at $J=7$, the upper bounds of the reasonable length
of the secular polynomial (cca 12 pages), of the reasonable
calculation time (cca 15 minutes) and of the memory consumption (822
M). Even the printing of the function $f_7(t,v,\ldots)$ already
requires a separation into its polynomial numerator
 $$
-{t}^{6}+ \left( w+m+y+n+z+v \right) {t}^{5}+
 $$
 $$
 + \left(
-my-nw-nv-yv-ny+m -vz+y-mv-mw-wz+w-nz+n+z \right) {t}^{4}+
 $$
 $$
 +\left(
-wz-2\,nw-vz+vnz+n-my+ vny+y+z-2\,ny \right .+
 $$
 $$
 +\left . wnz-nz-2\,mw-mv+m+mvy-nv
\right) {t}^{3}+
 $$
 $$
 +\left( vnz-my+
wnz+n-2\,nw+m+vny+z-2\,ny-nv-nz-mv-mw \right) {t}^{2}+
 $$
 $$
 +\left(
n-nv-nw- ny+m-nz \right) t+n
 $$
and denominator
 $$
{t}^{5}+ \left( -y-z-w-m-n \right) {t}^{4}+ \left(
my+wz+ny-z+mw-y-m-n +nw+nz \right) {t}^{3}+
 $$
 $$
 + \left(
-m+my+mw-wnz-n+2\,ny+nw-z+nz \right) {t }^{2}+ \left( -m-n+nw+nz+ny
\right) t-n\,.
 $$
Thus, we gave up the continuation of the process. We decided not to
reconstruct the next pair of functions $B_3$ and $A_3$ because the
completion of the analysis, straightforward as it is, would
certainly exceed the capacity of our desktop computer.

\subsection{Broader context: open-quantum-system connections}

It is worth
emphasizing that the results of the detailed mathematical
analysis
of the particular model (\ref{douinf}) (which,
in the light of this analysis,
appeared to be exactly solvable) may be perceived
as having a broader impact and relevance for
non-Hermitian physics.

Let us, first of all,
recall the general comments on physical
contexts as given in Introduction.
For the sake of definiteness
we will now
restrict our attention just
to the domain of quantum physics.
Moreover, let us pick up just the
subdomain of the theory in which
people study the models which are characterized by the presence of
non-Hermitian quantum operators.
In this subdomain
one finds, e.g., the studies of open quantum systems
\cite{Moiseyev,Rotter} with which our present text shares interest
in
a constructive approach to the coexistence of bound states with
the scattering solutions.

It is probably necessary to emphasize that
although the latter coexistence is being studied via
different formalisms in the literature,
it seems truly challenging to search for
their {\em shared} mathematical
features as well as phenomenological predictions.

On the purely formal level the most obvious meeting point
between the two specific (viz., ${\cal PT}-$symmetric
and open-system) approaches to the
non-Hermitian quantum physics and phenomenology
may be seen in the use of the Feshbach's
concept of the so called model space \cite{Feshbach}.
This means that
a certain projection-operator-specified
subspace of the complete Hilbert space
is chosen as a starting point. In this subspace
the exact Hamiltonian is replaced by its (often called
``effective'', i.e., isospectral)
reduction $H_{eff}$ as sampled here by
Eq.~(\ref{maco}).

Naturally, the results of the comparison can only have a
qualitative character.
Still, it is worth emphasizing that such a comparison
seems nontrivial due to its consequences. First of all,
let us recall the
most important physical problem of the unitarity of the
scattering process. For open systems this problem may be found
clarified, e.g., in Ref.~\cite{Ele} where one works with
the Hamiltonians
describing, in general, unstable systems.
For their genuine non-Hermitian Hamiltonians the spectrum of
the energies is not real so that the experimentalists are
allowed to speak about
resonances.
For theoreticians, the closeness of the
concept to the
${\cal PT}-$symmetric models immediately emerges when
one admits that the ${\cal PT}-$symmetry becomes spontaneously
broken \cite{Carl} so that
at least some of the energies become complex.

From the point of view of the experiment the
complementarity of the theoretical perspectives
seems best clarified in the scattering dynamical regime.
Indeed, in the ${\cal PT}-$symmetric context
the paradoxes (cf. \cite{Jones}) are being resolved by
an effective ``smearing'' of the potentials \cite{scatt}.
In parallel, the phenomenological description of the
open systems leads to a very similar emergence of the concept of
the ``smearing''. Its physical
essence is slightly different -
it is intimately connected to the
Feshbach's elimination and
transfer
of the scattering wave functions to a ``reservoir''
which is ``infinitely extended''.
But the outcome is similar:
the fundamental phenomenon of an effective
``smearing'' appears very naturally.

It is not too surprising that
many differences between the ${\cal PT}-$symmetric and
open systems survive.
In the latter case, for example,
the complex part of the energy
(which just characterizes an extent of the
spontaneous breakdown of the ${\cal PT}-$symmetry
in the former case)
has an immediate measurable, experimental interpretation
\cite{Kuhl}.
In addition, the imaginary part of the energy
in the open system
can be also related to the
nonorthogonality of
the resonance
states \cite{Savin}.
At the same time,
whenever one
succeeds in avoiding the complicated dynamical regime
of resonances (i.e., firstly, whenever the width of the resonance
happens to
shrink to zero and, secondly, whenever
the resonance leaves, simultaneously,
the continuous part of the spectrum),
the present class of models (\ref{douinf})
enters the game.
It
offers an amazingly rich variety of possible dynamical scenarios.
Their description
uses a manifestly non-perturbative, simplified but
still exact, non-numerical picture
of the qualitative changes in the quantum evolution
during which a stable bound state emerges below the continuum.

\section{Summary\label{summary}}

The current growth of interest in ${\cal PT}-$symmetric quantum
systems may be characterized by the novelty of its potential
physical contents as well as by the not entirely standard nature of
the related mathematics. Both of these challenges motivated also our
present choice of the family of Hamiltonians (\ref{douinf}).

Our choice was inspired, first of all, by the role of the special
cases of Hamiltonians (\ref{douinf}) (with $a=a'$, etc) in the
constructive disproof \cite{scatt} of the Jones' oversceptical
conjecture \cite{Jones} that the {\em scattering\,} by a generic
{\em local\,} ${\cal PT}-$symmetric potential $V(x)$ {\em cannot\,}
be unitary. For our present purposes it was important that the
resolution of the problem (cf.~also \cite{Jonesdva} and
\cite{discrete}) was provided by the replacement of the
``unsuitable'' local potentials (i.e., in our present
discrete-coordinate approach, of the diagonal matrices $V$) by their
``minimally non-local'' (i.e., tridiagonal-matrix) generalizations
as sampled by Eq.~ (\ref{potecy}).

Another important physical reason for the study of the latter class
of the ``smeared'' \cite{scatt} {\it alias} ``nearest-neighbor''
interactions was found in their level-attraction role played, in
particular, by their extremely non-Hermitian antisymmetric-matrix
special cases (with $a=a'$, etc) in the models of quantum
catastrophes \cite{catast}. In the latter case, the attention (cf.,
e.g., \cite{maximal,tridiagonal} remained restricted to the confined
motion (i.e., to the purely bound-state models) using not only
finite $J<\infty$ but also the truncated, finite-dimensional kinetic
energy operators (\ref{kine}). Hence, it was entirely natural for us
to ask what would happen when the truncation of the discrete
Laplacean (i.e., in effect, an external confining square-well-type
potential) were removed.

{\it A priori\,} we expected that the doubly infinite ${\cal
PT}-$symmetric toy-model matrices (\ref{douinf}) could enrich the
slowly developing classifications of quantum catastrophes
\cite{sdenisem} by a new, multiparametric class of models in which
the scattering dynamical regime could coexist with the parallel
(and, due to the non-Hermiticity of $H$, fragile) existence of bound
states (cf.~also \cite{Jonesdva} where more or less the same
questions were addressed via perturbation theory).

Our expectations proved confirmed. We found that in our model the
``birth'' of the stable bound states at the lower boundary of the
continuous spectrum may be rather easily controlled.
Serendipitously, we also discovered that besides its appeal in
physics, the model offers an unexpectedly long list of certain truly
amazing {\em formal} merits.

\begin{itemize}

\item
from the point of view of the Feshbach's constructive recipe
\cite{Feshbach} the description of all of our $J \leq \infty$
models (\ref{douinf}) proved non-numerical and reducible to the
study of the closed-form effective Hamiltonians $H_{eff}$
(cf.~Eq.~(\ref{maco})).

\item
the energy-dependence of the latter operators is rather complicated
in general. In our model it remained, in effect, elementary (i.e.,
basically, polynomial -- see, e.g., the six-parametric secular
Eq.~(\ref{secu3}) for illustration).

\item
a pleasant surprise appeared after our replacement of the search for
the bound-state energies $E_n$ as functions of the ``reduced''
couplings $u$, $v$, \ldots, by the search for the ``Sturmian''
bound-state couplings $u_n$ as functions of the freely variable
energy $E$ and of the remaining couplings $v$, $w$, \ldots. What was
a truly unexpected byproduct of this change of perspective was the
second-power-type factorization of polynomials leading to an
enormous simplification of the formulae.

\item
once we moved to the higher integers $J$ we imagined that the
ultimate, most natural version of the construction of the bound
states can be based on the last important simplification of their
localization using the continued-fraction re-arrangement of the
recipe. In the light of this result one could even contemplate a
transition to the models with $J =\infty$, in principle at least.

\item
the latter two forms of simplification made the bound-state
localization amenable also to its deeply intuitive and transparent
graphical (re-)interpretation. In this way our ``benchmark-model''
samples of the birth of the bound states from the continuum could
acquire a new physical relevance as the prototypes of quantum
catastrophes of a new, less usual type.

\item
for the growing integers $J$, the increase of the complexity of the
formulae caused by the linear growth of the number of the
independently variable parameters $u$, $v$, \ldots proved at least
partially compensated by the survival of their sufficiently
efficient and quick tractability by the symbolic-manipulation
software. The limitations of this approach only started emerging at
the values of $J$ as high as $J =8$. Still, one can hardly find too
many solvable quantum models with such a large number of
independently variable parameters.

\item
we should not forget to mention that the determination of the
$2J-$dimensional physical-parameter domains ${\cal
D}(a,a',b,b',\ldots)$ and/or of their boundaries $\partial {\cal D}$
(sampled, at $J=1$, by Fig.~\ref{fiut}) proved reducible to the mere
$J-$dimensional physical-parameter domains ${\cal D}(u,v,\ldots)$ as
sampled, at $J=2$, in Fig.~\ref{fihr}. Let us add that this was one
of the most unexpected discoveries concerning the potential physics
and phenomenology behind the model. The feature may be also
perceived, in retrospective, as one of the most important
{non-mathematical}, {\em descriptive\,} reasons of our interest in
the model in question.

\item
last but not least, we must emphasize that the most basic role was
played again by our acceptance of the enormously productive
\cite{Carl} requirement of ${\cal PT}$ symmetry
(cf.~Eq.~(\ref{syst})), i.e., in our model, of the symmetry of our
real matrix (\ref{douinf}) with respect to its antidiagonal.

\end{itemize}

\subsection*{Acknowledgements}

Work supported by the GA\v{C}R Grant No.~16-22945S.



\end{document}